\documentclass[runningheads]{llncs}
\usepackage[T1]{fontenc}
\usepackage{graphicx}
\usepackage{amsmath}
\usepackage{amssymb}
\usepackage{xspace}
\usepackage{cite}

\usepackage{xcolor}

\usepackage{amsthm}

\newcommand{\N}{\ensuremath{\mathbb{N}}\xspace}
\newcommand{\Z}{\ensuremath{\mathbb{Z}}\xspace}
\newcommand{\R}{\ensuremath{\mathbb{R}}\xspace}

\newcommand{\structure}[1]{\ensuremath{\left\langle#1\right\rangle}\xspace}
\newcommand{\ie}{\emph{i.e.}\@\xspace}
\newcommand{\set}[1]{\ensuremath{\left\{#1\right\}}}
\newcommand{\setofstates}{\ensuremath{S}\xspace}
\newcommand{\lattice}{\ensuremath{L}\xspace}
\newcommand{\neigh}{\ensuremath{N}\xspace}
\newcommand{\locrule}{\ensuremath{\delta}\xspace}
\makeatletter
\newcommand{\glorule}[1][\@nil]{%
\ensuremath{\def\tmp{#1}%
\ifx\tmp\@nnil{G_\locrule}\else{G_{#1}}%
\fi}\xspace}
\newcommand{\MSL}{\ensuremath{\mathcal{M}(\setofstates)^\lattice}\xspace}

\begin{document}
\title{Cellular Automata on Probability Measures}
%
%
\author{Enrico Formenti \inst{1} \and
Faizal Hafiz\inst{2} \and
Amelia Kunze\inst{1} \and
Davide La Torre\inst{2}}
\authorrunning{E. Formenti et al.}

\institute{Universit\'{e} C\^{o}te d'Azur, CNRS, i3S, Nice, France 
\and
SKEMA Business School, Sophia Antipolis, France}
\maketitle              
\begin{abstract}
Classical Cellular Automata (CCAs) are a powerful computational framework widely used to model complex systems driven by local interactions. Their simplicity lies in the use of a finite set of states and a uniform local rule, yet this simplicity leads to rich and diverse dynamical behaviors. CCAs have found applications in numerous scientific fields, including quantum computing, biology, social sciences, and cryptography. However, traditional CCAs assume complete certainty in the state of all cells, which limits their ability to model systems with inherent uncertainty. This paper introduces a novel generalization of CCAs, termed Cellular Automata on Measures (CAMs), which extends the classical framework to incorporate probabilistic uncertainty. In this setting, the state of each cell is described by a probability measure, and the local rule operates on configurations of such measures. This generalization encompasses the traditional Bernoulli measure framework of CCAs and enables the study of more complex systems, including those with spatially varying probabilities. We provide a rigorous mathematical foundation for CAMs, demonstrate their applicability through concrete examples, and explore their potential to model the dynamics of random graphs. Additionally, we establish connections between CAMs and symbolic dynamics, presenting new avenues for research in random graph theory.  This study lays the groundwork for future exploration of CAMs, offering a flexible and robust framework for modeling uncertainty in cellular automata and opening new directions for both theoretical analysis 
and practical applications.

\keywords{Cellular Automata  \and Probability Measures \and Random Graphs}
\end{abstract}

\section{Introduction}

Classical Cellular Automata (hereafter referred to as CCAs) are extensively used as formal tools to model various complex phenomena arising from local interactions and a finite number of states. CCAs find applications across numerous scientific domains, including quantum computing, biology, social sciences, and more~\cite{arrighi2019,VANSCOY201758,SIPAHI2018151,ZHAO2021126049}. In computer science, they are particularly valuable in cryptography, where they are employed for tasks such as pseudo-random number generation, secret sharing schemes, and the design of S-boxes~\cite{FormentiIMY14,MariotGFL20,PicekM0JM17}.

Informally, a CCA consists of an infinite collection of automata (referred to as cells), each of which assumes a state (or a value from the alphabet) drawn from a finite set (the set of states). These cells are arranged on a regular grid (the lattice), typically indexed by $\mathbb{Z}^d$, where $d$ denotes the dimensionality of the CCA. Each cell updates its state synchronously based on its own state and the states of neighboring cells (its neighborhood), using a uniform local rule that applies to all cells. The overall state of the automata at a given moment is referred to as a (instantaneous) configuration.

Despite the simplicity of the CCA model, it exhibits a remarkable diversity of dynamical behaviors. To study these behaviors, the set of configurations is equipped with the product topology (as detailed in the next section). This area remains an active field of research, leading to numerous publications in international conferences and journals.

Over time, researchers have proposed several variants of CCAs, some of which have inspired this work. For example, Cattaneo et al. introduced a model of cellular automata within fuzzy frameworks to better analyze the chaotic behavior of CCAs~\cite{CATTANEO1997105}. This foundational work has spurred a series of follow-up studies with practical applications, such as modeling crowd dynamics~\cite{GERAKAKIS2019125}.

In~\cite{CATTANEO1997105}, the authors justify their approach by noting that, in a configuration, only a finite subset of cells has a known state, while the states of the other cells are subject to uncertainty. Our paper proposes an alternative method for incorporating uncertainty into the local state. Specifically, we assume that the state of each cell is described by a probability measure which may vary from cell to cell. By imposing additional constraints on the structure of the local rule, we generalize CCAs to work with configurations of probability measures on the set of states. We term this new model Cellular Automata on Measures (CAMs).

This approach generalizes the conventional scenario in CCAs, where each cell is associated with the same Bernoulli measure. This natural extension has already led to several intriguing results; we refer to~\cite{Pivato12} for a survey of related findings. Our paper aims to provide a foundation for further exploration of these ideas. To highlight the breadth of this new model, we demonstrate how CCAs can be adapted to analyze the dynamics of random graphs, thereby establishing a connection between random graphs and symbolic dynamics.

\section{Classical Cellular Automata: Definitions and Main Properties}

Formally, a \textit{classical cellular automaton} (CCA or simply CA) is a structure \structure{d,r,S,\delta} where $d$ is the \emph{dimension}, \ie, the dimension
of the regular lattice \lattice (in this work we set either $\lattice=\N^d$ or $\lattice=\Z^d$); $r$ is the \emph{radius}; \setofstates is the finite set of \emph{states}; 
and  $\locrule\colon\setofstates^{(2r+1)^d}\to\setofstates$
is the \emph{local rule}. 
We denote by $N_r(i)$ the closed ball of center $i\in\lattice$ and radius $r\in\N$ according to the standard Manhattan distance on \lattice. 
We call $N_r(i)$ the \emph{neighborhood} of the cell $i$ (on \lattice).
Any element of $x\in S^\lattice$ is called a \emph{configuration} and it
represents the overall state of the CA at a given generic time step.
For any $i\in\lattice$, denote by $x_i=x(i)$ the element of $x$ contained in the cell $i\in L$. This notation can be extended to (ordered) sequences of elements of \lattice in the obvious way. In particular, $x_{\neigh_r(i)}=x(N_r(i))$ denotes the content of the neighborhood of cell
$i$ for the configuration $x$.
Any local rule \locrule induces a \emph{global rule} \glorule as follows
\[
\glorule(x)_i=\locrule\left(x_{\neigh_r(i)}\right), \ \ \
\forall x\in\setofstates^\lattice, \  \forall i\in\lattice.
\]
Hence, the global rule illustrates the overall
evolution of current configuration of the cellular automaton after one
unity of time. 

In order to study CCA as
dynamical systems, the set of
states \setofstates is equipped with
the topology induced by the discrete metric $d_D$ defined
as follows

\begin{equation*}
d_D(s,t) =
\begin{cases}
     1 & s\ne t \\
     0 & \text{otherwise}
\end{cases}
\end{equation*}
for all $s,t\in\setofstates$. Then, 
$\setofstates^\lattice$ is endowed with \emph{Cantor topology} \ie
the standard product
topology induced by the discrete topology on \setofstates.
Consider the
\emph{Cantor distance} defined
as:
\begin{equation*}
d_C(x,y)=
\sum_{k=-\infty}^{+\infty} \frac{d_D(x_k,y_k)}{s^{|k|}}   
\end{equation*}
where $s$ is the size of \setofstates
(recall that \setofstates is a finite set here).
It is well-known that the
Cantor distance induces the
Cantor topology on $\setofstates^\lattice$.






For a fixed $j\in \lattice$, the \emph{shift map} $\sigma_j$ 
is a very well-known and widely studied CA defined as follows 
\[
\sigma_j(x)_i=x_{i+j}, \ \ \forall x\in\setofstates^\lattice,\ \forall i\in\lattice.
\]
%
The famous Curtis-Hedlund-Lyndon theorem established that CCAs are exactly those continuous maps commuting with the shift and
can be formally stated as follows.
\begin{theorem}\cite[Th. 3.4 p. 324]{hedlund69}
    The class of CCAs are exactly the class of continuous mappings $f\colon\setofstates^\lattice\to\setofstates^\lattice$ such that for all $j\in\N$, $f\circ\sigma^j=\sigma^j\circ f$.
\end{theorem}

\section{The Space of Probability Measures}
\label{proba}

In the following let us suppose that the state space $S$ is endowed with a metric $d$ and that $(S,d)$ is a compact metric space. The space $S$ can be discrete as $\{0,1\}$ or continuous as $[0,1]$. Let $\mathcal{B}(S)$ be the Borel sigma-algebra defined on $S$. Let us denote by $\mathcal{M}(S)$ the set of all probability measures defined on $S$. $\mathcal{M}(S)$ can be endowed with the Monge-Kantorovich metric defined as:

\begin{equation}
d_{MK}(\mu,\nu) = \sup_{f\in Lip_1(S)} \int_S f d\mu - \int_S f d\nu     
\end{equation}

where $Lip_1(S)$ is the set of all Lipschitz-1 functions on $S$, that is 

\begin{equation}
Lip_1(S) = \{f:S\to\R\text{ s.\@\,t.\@\;} |f(x)-f(y)|\le d(x,y)\}.   
\end{equation}

{\color{black}
The Monge-Kantorovic metric arose from a classical problem in transportation of mass \cite{Mendivil2015}. 
Since then it has been employed in many applications areas, several of which are stochastic in nature and require a means to measure the distance between two probability measures. 
Indeed, the set of probability measures can be equipped with the Monge-Kantorovich metric in order to yield the topology of weak convergence \cite{Durrett1992}.
A particularly important feature of the Monge-Kantorovich distance is that it 
is tied to the underlying metric $d$ on the space $S$. A simple example illustrating this link is provided below. 
}

\begin{proposition}
Let us suppose that $S$ is a finite and discrete set. Let $d$ the discrete metric defined on $S$, and let us consider the set $D\subseteq \mathcal{M(S)}$ 
composed by all Dirac measures on $S$, that is
\begin{equation}
\mathcal{D} = \{\delta_s, s\in S \}.    
\end{equation}
Then the metric $d_{MK}$ collapses to the discrete metric $d$ on $\mathcal{D}$.
\end{proposition}

\begin{proof}
Let us consider two points $s$ and $t$ in $S$ and the associated Dirac probability measures $x=\delta_s$ and $y=\delta_t$.
Computing $d_{MK}(x,y)$ yields
\begin{equation}
d_{MK}(x,y) = \sup_{f\in Lip_1(S)} f(s) - f(t) = d_D(s,t) ,
\end{equation}

that is, the $d_{MK}$ metric collapses to the discrete metric $d$.
\end{proof}

{\color{black}

The following result shows that the metric $d_{MK}$ may be easily computed for real-valued probability measures. 

\begin{proposition}\cite{vallender1974}
    Let $\mu$ and $\nu$ be two probability measures on $S \subset \mathbb{R}$ and denote their cumulative distribution functions by $F_\mu$ and $F_\nu$, respectively. Then the metric $d_{MK}$ can be expressed as
    \begin{equation}
        d_{MK}(\mu, \nu) = \int_\mathbb{R} \lvert F_\mu - F_\nu \rvert dx.
    \end{equation}
\end{proposition}

Note that $d_{MK}$ simplifies even further in the case of Bernoulli probability measures. Let $\mu$ be a Bernoulli measure with probability $q$ and $\nu$ be a Bernoulli measure with probability $p$. Then,
\begin{equation}
    d_{MK}(\mu, \nu) = \int_\mathbb{R} \lvert F_\mu - F_\nu \rvert dx =  \ \int_0^1 \lvert (1-q) - (1-p) \rvert dx = \lvert p-q\rvert.
    \label{eqn:L1}
\end{equation}
In words, we obtain the desirable property that the Monge-Kantorovich distance simplifies to the $L^1$ distance between the probabilities $p$ and $q$.
}

\begin{remark}
Let us notice that for 
$f\in Lip_1(S)$ and 
$0\in S$ it holds

\begin{equation}
\int_S f d\mu - \int_S f d\nu  \le \int_S f(0)+ d(x,0) d\mu - \int_S f(0)-d(x,0)d\nu \le 2 K  
\end{equation}

where $K=\max_{x\in S} d(x,0)$. In other words, $d_{MK}(\mu,\nu)\le 2K$ for any $\nu,\mu\in \mathcal{M}(S)$, and hence $(d_{MK}, \mathcal{M(S)})$ is a bounded metric space.
\end{remark}

\begin{proposition}
\cite{hutchinson1981fractals}
If $(S,d)$ is a compact metric space then
the metric space $(\mathcal{M}(S),d_{MK})$ is compact and, therefore, complete.
\end{proposition}

Let us notice that the space $\mathcal{M}(S)$ is not a group but it is closed with respect to the convex combination of probability measures (see \cite{barnsley2014fractals, hutchinson1981fractals,kunze2011} for more details).

\begin{proposition}
\cite{barnsley2014fractals} 
Let $\mu_1,\mu_2,...,\mu_n\in \mathcal{M}(S)$ and $\nu_1,\nu_2,...\nu_n\in \mathcal{M}(S)$ be two sets of probability measures and $\lambda_1, \lambda_2, ...,\lambda_n$ a set of weights such that $\sum_{j=1}^n \lambda_j = 1$. Then, the following inequality holds:

\begin{equation}
d_{MK} \left(\sum_{j=1}^n \lambda_j \mu_j,\sum_{j=1}^n \lambda_j \nu_j\right)\le \sum_{j=1}^n \lambda_j d_{MK}(\mu_j,\nu_j)   
\end{equation}

\end{proposition}

\section{Cellular Automata on Probability Measures}

A Cellular Automaton on probability Measures (CAM) is defined on configurations whose elements are probability measures. In other words, in this new setting, the standard state set $S$ is replaced with $\mathcal{M}(S)$. 
Let us notice that this definition collapses to the classical one when the probability measure in the cell $i$ coincides with the Dirac measure on a specific state $s\in S$. In other words, a string of states for a CCA can be embedded into this new stochastic framework by associating with each cell element $s\in S$ the probability mass at $s$.

{\color{black}
In the classical setting where $S$ is a discrete set of states, $S^L$ is not a convex space. A nice property of our framework, as we prove below, is that $\MSL$ is closed with respect to convex combinations of configurations. This fact motivates a particularly interesting choice of local rule which we will define in the next section.  
}

We are now ready to introduce the distance $d_\mathcal{M}$ between two different configurations $x$ and $y$ in $\MSL$, which relies on the Monge-Kantorovich distance $d_{MK}$ between the measures $x_i$ and $y_i$ defined in each cell.

\begin{definition}
For any pair of elements $x,y\in \mathcal{M}(\setofstates)^\lattice$, let us define the function

\begin{equation}
d_{\mathcal{M}}(x,y) = 
\sum_{i\in\lattice} \frac{d_{MK}(x_i,y_i)}{2^{|i|}}    
\end{equation}
where $|\cdot|$ is the norm of $i$.
\end{definition}

We now present a simple example illustrating this new framework.
\begin{example}
 Consider $\setofstates=\set{0,1}$ and the set of all possible discrete probability measures on \setofstates. Any probability can take only two values $p$ and $1-p$, respectively. Let us now consider a set of possible configurations in $\setofstates^\Z$ and, for the sake of simplicity, consider the elements $x_{-1}$, $x_0$ and $x_1$ at the time $t$. We now consider the transition map that associates each triple as above with a new probability measure defined by the convex combinations of $x_{-1}$, $x_0$, and $x_1$ with weights $\lambda_{-1}$, $\lambda_0$, and $\lambda_1$, respectively. An illustration is provided in Fig. \ref{CAbasicgraph} with $\lambda_{-1}=0.2$, $\lambda_0=0.4$, and $\lambda_{1}=0.4$. 
Let us notice, using Eq.~(\ref{eqn:L1}), that in this case the CA $(\setofstates^\Z,d_{\mathcal{M}})$ is 
equivalent to the CA $([0,1]^\Z,d^*_{L^1})$ where $d^*_{L^1}$ is the weighted $L^1$ distance defined as;
$$
d^*_{L^1}(x,y)= \sum_{k=-\infty}^{+\infty} \frac{|x_k-y_k|}{2^{|k|}}
$$\qed
\end{example}

\begin{figure}[t]
\begin{center}
\includegraphics[width=0.67\textwidth]{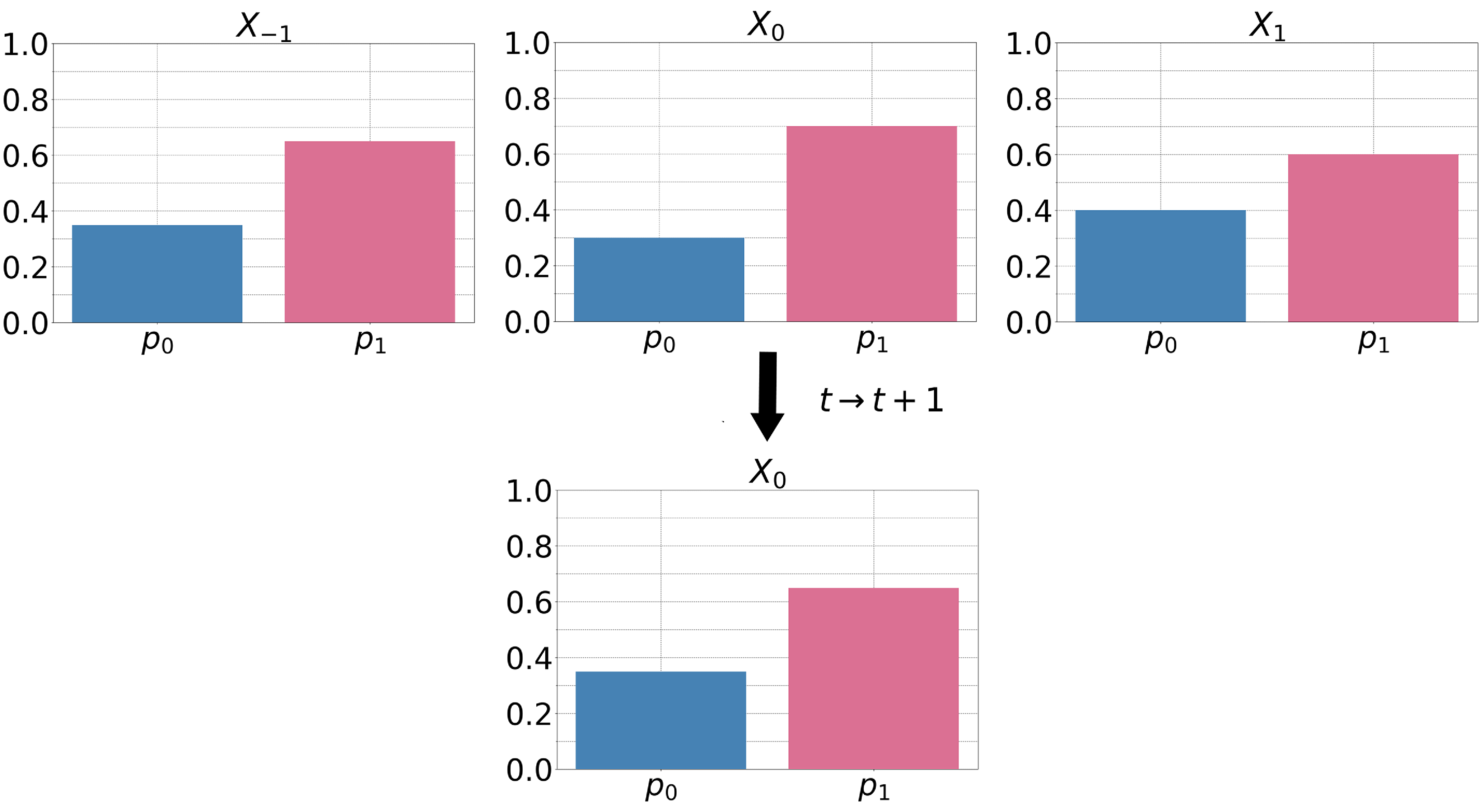}
\end{center}
\caption{CA on Probability Measures defined on \set{0,1}}
\label{CAbasicgraph}
\end{figure}





The following proposition proves some topological properties of the space
$(\mathcal{M}(\setofstates),d_{\mathcal{M}})$.

\begin{proposition}
\label{prop:MSL-compact}
Suppose that $(\setofstates,d)$ is compact. Then, $(\mathcal{M}(\setofstates)^\lattice,d_{\mathcal{M}})$ is compact and, therefore, complete. 
\end{proposition}
\begin{proof}
The compactness of \setofstates implies that the above function $d_{\mathcal{M}}$ is well-defined and bounded.
Furthermore, $d_{\mathcal{M}}$ generates the product topology \cite{sagar2021compactness}.
The remainder of the proof follows from Tychonoff's theorem which states that any countable product of compact metric spaces is compact in the product topology~\cite[Th. 13 p. 143]{kelley1975}. Completeness follows immediately from compactness. 
\end{proof}

\begin{proposition}
\label{prop:MSL-convex}
$\mathcal{M}(S)^\lattice$ is a convex space.
\end{proposition}

\begin{proof}
To show the convexity property, let $\lambda\in [0,1]$ and $x,y\in \mathcal{M(S)}^L$. By trivial calculations it is easy to prove that $\lambda x + (1-\lambda)y \in \mathcal{M(S)}^L$.
\end{proof}




\begin{proposition}\label{prop:f-continuous-then-has-fix-point}
Let \setofstates be a compact set, and $f:(\mathcal{M(\setofstates)}^\lattice,d_{\mathcal{M}})\to (\mathcal{M(\setofstates)}^\lattice,d_{\mathcal{M}})$ be a continuous function. Then $f$ has at least one fixed point.    
\end{proposition}

\begin{proof}
By Propositions~\ref{prop:MSL-compact} and~\ref{prop:MSL-convex},
$(\mathcal{M}(\setofstates)^\lattice,d_{\mathcal{M}})$ is compact and convex. Since $f$ is continuous, the thesis follows from Schauder's fixed point theorem~\cite{schauder1930fixpunktsatz,zeidler2012applied}.
\end{proof}
\smallskip

The following result shows that the shift map is Lipschitz. 

\begin{proposition}
For any $j\in \lattice$, the shift map  $\sigma_j(x)_i=x_{i+j}$ is Lipschitz continuous on 
$(\mathcal{M(S)}^L,d_{\mathcal{M}})$.
\end{proposition}

\begin{proof}
To prove it, let us take two elements $x$ and $y$ in $\mathcal{M(S)}^L$ and calculate the distance between $\sigma_j(x)$ and $\sigma_j(y)$. Computing, we have:
\begin{align}
d_{\mathcal{M}}(\sigma_j(x),\sigma_j(y)) &=
\sum_{i\in\lattice} {d_{MK}((\sigma_j(x))_i,(\sigma_j(y))_i)\over 2^{|i|}} \nonumber\\ 
&=
\sum_{i\in\lattice} {d_{MK}(x_{i+j},y_{i+j})\over 2^{|i|}}
\end{align}
and by doing the substitution $k=i+j$ and remembering that $|j-k|=|k-j|\ge |k| - |j|$ we have that
\begin{align}
d_{\mathcal{M}}(\sigma_j(x),\sigma_j(y)) &=
\sum_{k\in\lattice} {d_{MK}(x_{k},y_{k})\over 2^{|j-k|}} \le 
\sum_{k\in\lattice} {d_{MK}
(x_{k},y_{k})\over 2^{|k|-|j|}} \nonumber\\[2mm] 
&=
2^{|j|} \sum_{k\in\lattice} {d_{MK}(x_{k},y_{k})\over 2^{|k|}} = 2^{|j|} d_{\mathcal{M}}(x,y)
\end{align}
\end{proof}
\smallskip


It is easy to prove that the shift map $\sigma_j$ has an uncountable set of fixed points, namely all constant sequences taking the form $x_i=\mu$ for any probability measure $\mu\in \mathcal{M}(\setofstates)$.

The following result presents the case of a contraction mapping $f$ on \MSL.

\begin{proposition}
Let $f:\MSL\to\MSL$ be a contraction mapping, that is there exists a constant $c_f\in [0,1)$ such that:
\begin{equation}
d_{\mathcal{M}}(f(x),f(y))\le C d_{\mathcal{M}}(x,y)    
\end{equation}
Then, $f$ has a unique fixed point $\bar x\in \mathcal{M(S)}^\lattice$, and for any configuration $x^0\in \mathcal{M(S)}^\lattice$ the orbit $f^n(x^0)$ converges to $\bar x$ whenever $n\to +\infty$. 
\end{proposition}

\begin{proof}
The proof follows by applying Banach's fixed point theorem \cite{kunze2011} to the complete metric space $\mathcal{M(S)}^\lattice$. 
\end{proof}

\section{A Local Rule: The Convex Combination Map}
\label{sec:convex-combination-map}

In this section we explore the properties of a specific local rule $\delta_C$ which is defined as follows.
For any $i \in \lattice$ and all $j\in N_r(i)$, let $\{\lambda_j\}$ be a set of weights in $[0,1]$ such that $\sum_{j\in N_r(i)} \lambda_j = 1$. 
Let us notice that the set of weights $\lambda_j$ is only dependent on the position of the element $j\in N_r(i)$ and is not dependent on $i$.  
Let $T: \mathcal{M(S)}\to \mathcal{M(S)}$ be a map defined on the space of probability measures $\mathcal{M(S)}$.
For any $x\in \MSL$ the action of $\delta_C$ on $x$ is defined as 
\begin{equation}
(\delta_C(x))_i = \sum_{j\in N_r(i)} \lambda_j T(x_j)    
\label{eqn:ccrule}
\end{equation}


{\color{black}
Note that the shift map $\sigma_j(x)_i$ is contained in the set of possible convex combination rules expressible in the form of Eq.~(\ref{eqn:ccrule}). To see this, take the identity map $T(x)=x$ with only one non-zero coefficient $\lambda_{i-j}=1$.
}


{\color{black}
As we show below, an interesting consequence of this framework is that the Lipschitz property of the map $\delta_C$ depends entirely on the Lipschitz property of the map $T$. The remainder of this section establishes such Lipschitz properties and provides convergence results in the case of contractivity. 
}

\begin{proposition}
Suppose that the map $T$ satisfies the following average Lipschitz property: There exists a constant $c$ such that for any configurations $x,y \in \MSL$ and any $i\in \lattice$ it holds:

\begin{equation}
\sum_{j\in N_r(i)} \lambda_j d_{MK}(T(x_j),T(y_j))\le c d_{MK}(x_i,y_i)    
\end{equation}
Then the map $\delta_C:\MSL\to\MSL$ is a Lipschitz map with respect to $d_{\mathcal{M}}$ with Lipschitz factor $c$. 
\label{proofbanach1}
\end{proposition}

\begin{proof}
In order to prove this result, let us take two elements $x,y\in\MSL$ and calculate the $d_{\mathcal{M(S)}}$ distance between them. Computing, we get:    

\begin{align}
d_{\mathcal{M}}(\delta_C(x),\delta_C(y)) &=  
\sum_{i\in\lattice} {d_{MK}((\delta_C(x))_i,(\delta_C(y))_i)\over 2^{|i|}} \nonumber\\
&=
\sum_{i\in\lattice} {d_{MK}\left(
\sum_{j\in N_r(i)} \lambda_j T(x_j),\sum_{j\in N_r(i)} \lambda_j T(y_j)
\right)\over 2^{|i|}} \nonumber\\
&\le
\sum_{i\in\lattice} \sum_{j\in N_r(i)} \lambda_j {d_{MK}\left(T(x_j),T(y_j)
\right)\over 2^{|i|}}\nonumber\\[2mm]
&\le c d_{\mathcal{M}}(x,y)
\end{align}
\end{proof}

\begin{remark}
Let us notice that the above average Lipschitz property collapses to the classical notion whenever $r=0$ and $N_r(i)=\set{i}$ for all $i\in L$.     
\end{remark}

\begin{proposition}
Suppose that the map $T$ satisfies the above average Lipschitz property with Lipschitz factor $c\in [0,1)$. Then the map $\delta_C: \mathcal{M(S)}^L \to \mathcal{M(S)}^\lattice$ is a contraction with contractivity factor $c$. Furthermore, there exists a unique invariant configuration $\bar x$ such that $\delta_C(\bar x)=\bar x$ and, for any configuration $x_0$, the orbit $\delta_c^n (x_0)$ converges to $\bar x$ whenever $n\to +\infty$.   
\end{proposition}

\begin{proof}
The proof follows from first applying Proposition~\ref{proofbanach1} and then Banach's fixed point theorem to the contraction $\delta_c$ in the complete metric space $\MSL$.
\end{proof}

\begin{proposition}
Suppose that the map $T$ satisfies the following cell-dependent Lipschitz property: There exists a constant $c$ such that for any $i\in \lattice$ the following condition holds: 
\begin{equation}
d_{MK}(T(x_i),T(y_i))\le \frac{c}{2^{|i|}} d_{MK}(x_i,y_i)    
\end{equation}
Then, the map $\delta_C: \mathcal{M(S)}^L \to \mathcal{M(S)}^\lattice$ is Lipschitz with respect to $d_{\mathcal{M}}$ with Lipschitz constant $C=\sum_{i\in\lattice} 
\frac{c}{2^{|i|}}$.
\label{proofbanach2}
\end{proposition}

\begin{proof}
In fact, similarly to the proof of the previous proposition, we have:
\begin{align}
d_{\mathcal{M}}(\delta_c(x),\delta_c(y)) &= \sum_{i\in\lattice} {d_{MK}((\delta_C(x))_i,(\delta_C(y))_i)\over 2^{|i|}} \nonumber\\
&= \sum_{i\in\lattice} {d_{MK}\left(
\sum_{j\in N_r(i)} \lambda_j T(x_j),\sum_{j\in N_r(i)} \lambda_j T(y_j)
\right)\over 2^{|i|}} \nonumber\\[1mm]
&\le
\sum_{i\in\lattice} \sum_{j\in N_r(i)} \lambda_j {d_{MK}\left(T(x_j),T(y_j)
\right)\over 2^{|i|}} \nonumber\\[1mm]
&\le
\sum_{i\in\lattice} \sum_{j\in N_r(i)} 
\frac{c \lambda_j}{2^{|i|} 2^{|j|}} 
d_{MK}(x_j,y_j) \nonumber\\[1mm]
&\le  
\sum_{i\in\lattice} 
\frac{c}{2^{|i|}} \left(
\sum_{j\in N_r(i)} 
\frac{\lambda_j}{2^{|j|}} d_{MK}(x_j,y_j) 
\right) \nonumber\\[2mm]
&\le 
\left(\sum_{i\in\lattice} 
\frac{c}{2^{|i|}}\right) d_{\mathcal{M}}(x,y)
\end{align}
\end{proof}


\begin{proposition}
Suppose that the map $T$ satisfies the above cell-dependent Lipschitz property with constant
$C=\sum_{i\in\lattice} \frac{c}{2^{|i|}}<1$. 
Then, the map $\delta_C:\MSL\to\MSL$ is a contraction with contractivity factor $c$. Furthermore, there exists a unique invariant configuration $\bar x$ such that $\delta_C(\bar x)=\bar x$ and, for any configuration $x_0$, the orbit $\delta_c^n (x_0)$ converges to $\bar x$ whenever $n\to +\infty$ in the $d_{\mathcal{M}}$ metric.   
\end{proposition}

\begin{proof}
The proof follows from first applying Proposition~\ref{proofbanach2} and then Banach's fixed point theorem to the contraction $\delta_c$ in the complete metric space $\MSL$.
\end{proof}

\section{Random Graphs and Two-Dimensional Cellular Automata}
\label{sec9}\label{sec:final-rem-random-graphs}

{\color{black}
It is well known that a random graph is a graph in which properties such as the number of vertices and edges, as well as the locations of those edges, are determined in some random way. 
In practical applications, random graphs are particularly useful for modeling dynamics in a human population. They can be used, for example, in social network analysis to describe the level of interaction between people within the same network.
In this setting, each person in the network is connected with all other members through edges. Associated with each edge connecting two members of the network described by the nodes $i$ and $j$ is a probability measure which models the level of intensity of their relationship. A scenario in which the two people have no relationship is described by the Dirac measure concentrated at $0$. 
Of course, personal interactions evolve over time and this will generate a dynamic behaviour. 

In this section, we will illustrate how a two-dimensional cellular automata on probability measures can be used to describe the dynamic behaviour of a random graph. We will first introduce a mathematical formalism which elucidates the link between CAMs and random graphs. We then provide a simple example to illustrate the potential of CAMs to model random graphs.
}

Consider a countable set of nodes $N_j$, $j \in \mathbb{Z}$, and let $M_{ij}$ denote the adjacency matrix of the graph, where $i, j \in \mathbb{Z}$. The matrix $M_{ij}$ can be either symmetric or asymmetric, depending on whether the graph is undirected or directed. For each entry of $M_{ij}$, we associate a probability measure on $[0,1]$ to represent the intensity of the link between nodes $i$ and $j$. This setup can naturally be described using the bidimensional CAMs introduced in previous sections. Here, the set of states is $S = [0,1]$ and the lattice is $\mathbb{Z}^2$, making the space of configurations $\mathcal{M}(S)^{\mathbb{Z}^2}$. Each configuration corresponds to the state of the random graph and encodes the probability of a specific link between nodes $i$ and $j$. 
Figure~\ref{CAgraph} provides a visual illustration showing the link between random graphs and their cellular automata representation.
For simplicity, we focus on the two-dimensional case, although extending the approach to higher-dimensional cellular automata is straightforward. 

Let us suppose we want to calculate the centrality index of a certain vertex within this framework. 
If we consider a node $i$ in the graph, we know that $x_{ij}\in \mathcal{M}(S)^{\Z^2}$, for any $j\in\N$, represents the probability associated with  the edge connecting the node $i$ and the node $j$. 

\begin{definition}
We say that a node $i$ has $\alpha$-centrality in a graph if for any node $j$ connected with $i$ we have that $x_{ij}([\alpha,1])\neq 0$. The $\alpha$-centrality index is then defined as
\begin{equation}
C_\alpha = \Pi_{j\neq i} x_{ij}([\alpha,1]).
\end{equation}
\end{definition}

\begin{figure}[t]
\begin{center}
\includegraphics[width=0.62\textwidth]{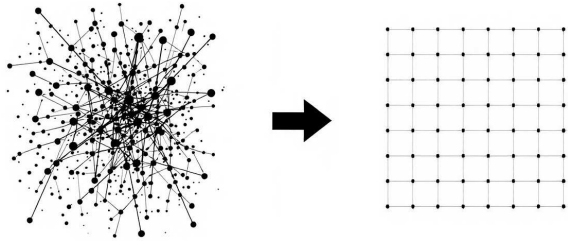}
\end{center}
\caption{CAs and Random Graphs. 
}
\label{CAgraph}
\end{figure}

Let us notice that the $\alpha$-centrality definition depends on the specific configuration of the cellular automata associated with the graph. Because any cellular automata on $\mathcal{M}(S)^{\Z^2}$ induces a dynamic on the space of all possible configurations, the notion of $\alpha$-centrality also depends on the particular state configuration. Let us also notice that the action at time $t+1$ of the cellular automata on the probability measure defined on a specific edge depends on the same probability as well as the convex combination of the probabilities of the surrounding edges at the time $t$. 
The CA operator acting on the space $\mathcal{M}(S)^{\Z^2}$ takes the form

\begin{equation}
(\delta_C(x))_i = \sum_{j\in N_r(i)} \lambda_j T(x_j)    
\end{equation}

where the operator $T:\mathcal{M}([0,1])\to \mathcal{M}([0,1])$ belongs to the families of mappings which have been described in the previous sections. The choice of the parameters $\lambda_j$ as well as the size of the neighborhood $N_r(i)$ will affect the dynamics induced by $\delta_C$ on the graph. 
The dynamics induced by the CA on the random graph will also affect the path connecting different nodes. 
Let us introduce the following definition.

\begin{definition}
Given two nodes $i$ and $j$ we say that $i$ and $j$ are connected through a random path if there exists a sequence of edges $\mathcal{P}$ connecting $i$ and $j$ such that $x_{sl}((0,1])\not=0$ for any $x_{sl}\in\mathcal{P}$. \end{definition}

In other words, two nodes $i$ and $j$ are connected through a random path $\mathcal{P}$ if the probability of being a positive number in $(0,1]$ is non-zero for each edge in the path. 


If the operator $T$ is a contraction, then the operator $\delta_C$ has a unique fixed point $\bar x$ in $\mathcal{M}([0,1])^\Z$ and, for any $i$ and $j$, $x_{ij}$ is constant and it is equal to the unique fixed point of $T$. In this case, no matter what is the initial distribution of probabilities on edges, the random graph will converge to a limit situation in which all edges will have the same distribution (which is the fixed point of the operator $T$). 
This is summarized by the following proposition.

\begin{proposition}
Let $x\in \mathcal{M}([0,1])^\Z$ be a configuration of the random graph. Let $\delta_C$ be the local map whose action is defined as:  
\begin{equation}
(\delta_C(x))_i = \sum_{j\in N_r(i)} \lambda_j T(x_j) 
\label{eqn:deltaC}
\end{equation}
where the map $T$ is a contraction. Then when $n\to +\infty$ the random graph will evolve towards an invariant configuration corresponding to the fixed point of $\delta_C$. 
\label{thm:example3}
\end{proposition}

In the rest of this section, we present a computational example to illustrate how dynamics on random graphs can be modeled using CAMs.
We perform the following computation in Maple.
Our CAM is represented as a matrix $M$ that has a large, but finite, dimension. Each cell $x_{ij}$ of $M$ contains a Bernoulli distribution with probabilities $p_{ij}$ and $1-p_{ij}$.  The real number $p_{ij} \in [0,1]$ models the strength of the link between node $i$ and node $j$ in the associated graph. Here, we consider an undirected graph so that the strength of the edge connecting node $i$ to node $j$ is equal to the strength of the edge connecting node $j$ to node $i$. This yields a symmetric matrix $M$ with probabilities $p_{ij} = p_{ji}$.

To generate dynamics on the random graph, we iteratively apply a local rule $\delta_C$ to each element of the matrix $M$. The local rule is defined by 
 \begin{equation}
(\delta_C(x))_i = \sum_{j\in N_r(i)} \lambda_j x_j
\label{eqn:example1_localrule}
 \end{equation}
 (this corresponds to $T: \mathcal{M}([0,1]) \to  \mathcal{M}([0,1])$ being the identity map and no longer a contraction in Eq.~(\ref{eqn:deltaC})). 
 We use $r=1$ so that $N_r(i)$ denotes the standard Von Neumann neighborhood. 
 The $\lambda_j$ are non-negative coefficients randomly generated 
 with sum equal to 1. In particular, letting $\lambda_{0,0}$ denote the coefficient for the central cell, we use 
 \begin{equation}
 \Lambda =  \begin{bmatrix}
    0 & \lambda_{1,0} & 0\\
    \lambda_{0,-1} & \lambda_{0,0} & \lambda_{1,0}\\
    0 & \lambda_{-1,0} & 0 
    \end{bmatrix} =
    \begin{bmatrix}
    0 & 0.5211 & 0\\
    0.4247 & 0.0049 & 0.0151\\
    0 & 0.0042 & 0\\
    \end{bmatrix},
    \label{eqn:example2_lambda}
 \end{equation}
 rounded to 4 decimal places.

We model a graph with $25$ nodes and hence require a weight matrix $M$ where $M \in [0,1]^{25 \times 25}$.  We initiate the weight matrix $M$ with random (Bernoulli) probabilities $p_{ij}$ in every cell $x_{ij}$. 
In order to update the elements along the border of $M$ (which, once again, must be finite in order to be implemented), we consider the lattice as a torus and implement a periodic boundary condition.



The results of the computation are provided in Figure~\ref{fig:example2}. Figure~\ref{fig:example2} (a) depicts the initial (randomized) state of the graph, denoted $G_0$. The darkness of each edge reflects the probability of connection between the corresponding vertices: the darker the edge, the closer the associated probability $p_{ij}$ is to $1$. Figure~\ref{fig:example2} (b) depicts the state after applying the global rule once. In general, we denote by $G_n$ the graph obtained after applying the global rule $n$ times. 

Note that, because we instantiate every cell in $M$ with a probability $p_{ij} \in [0,1]$, our initial graph allows for nodes to be directly connected to themselves by way of a single edge. While we could enforce zeroes along the diagonal of our initial matrix $M_0$, 
and hence remove these loops in our initial graph $G_0$, 
we cannot avoid creating these edges in subsequent graphs after applying the same local rule at every element of our CAM $M$. 

\begin{figure}[!t]
    \centering
    \begin{tabular}{cc}
    (a)\hspace{1mm} $G_0$ & (b)\hspace{1mm} $G_1$\\
    \includegraphics[width=0.36\textwidth]{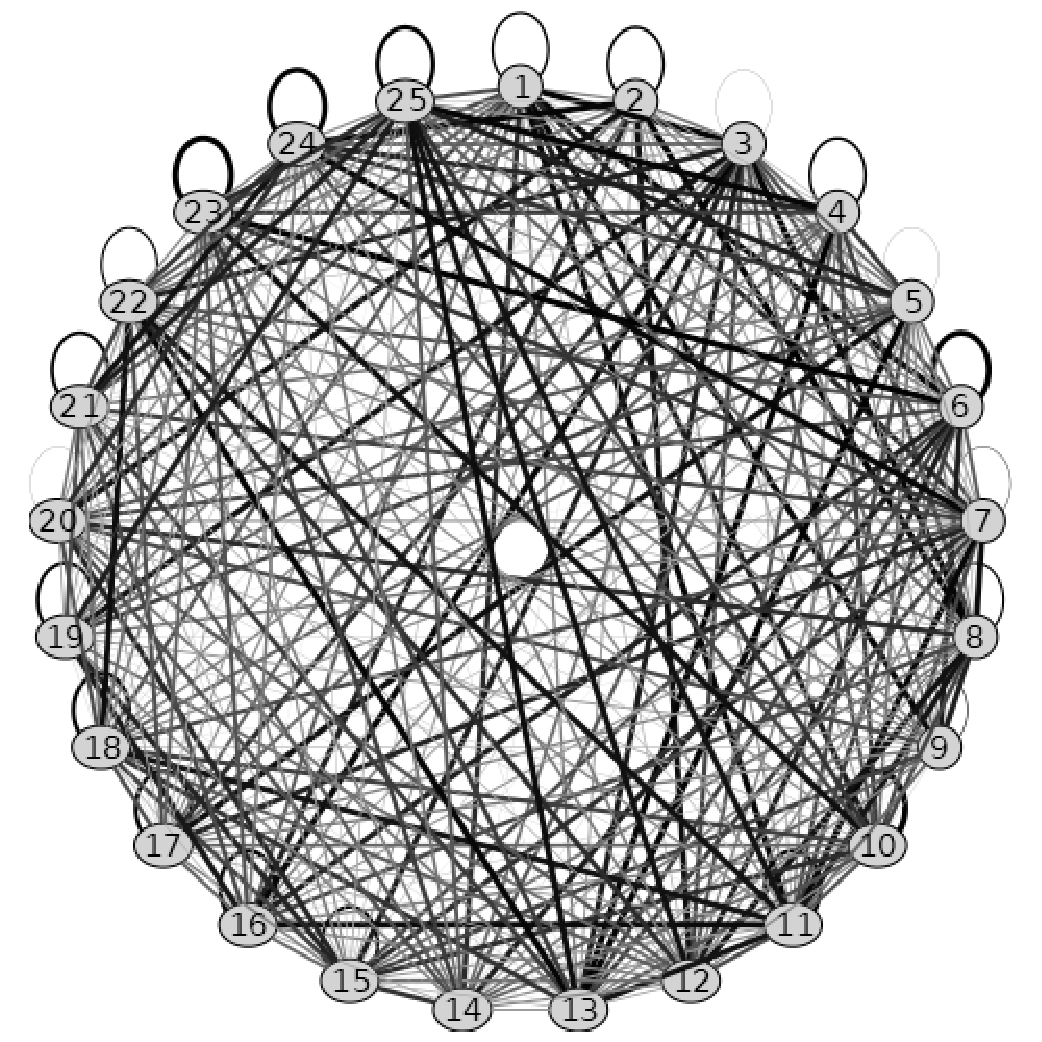} & 
    \includegraphics[width=0.36\textwidth]{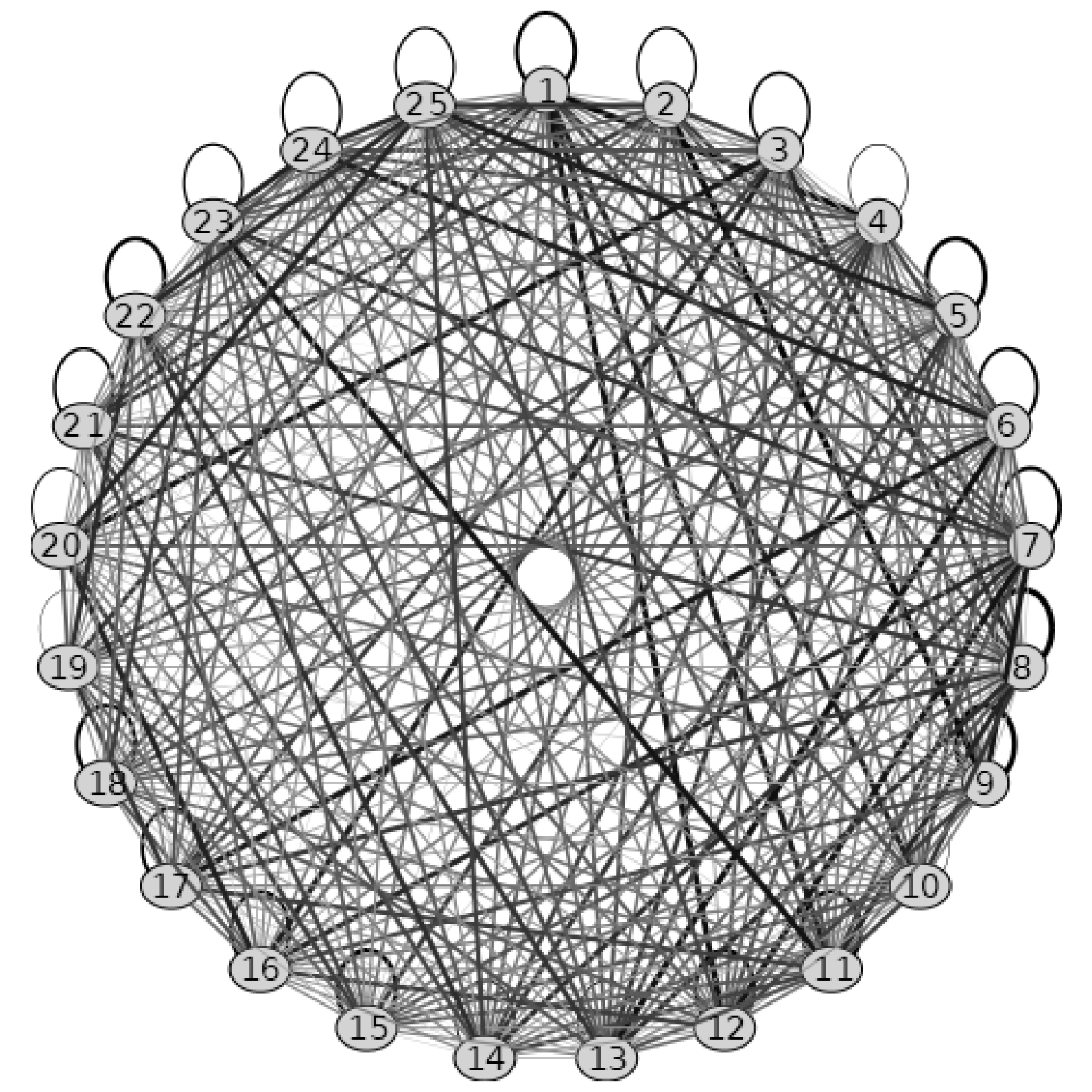}
     \\[4mm]
    (c)\hspace{1mm} $G_2$ & (d)\hspace{1mm} $G_{100}$\\
    \includegraphics[width=0.36\textwidth]{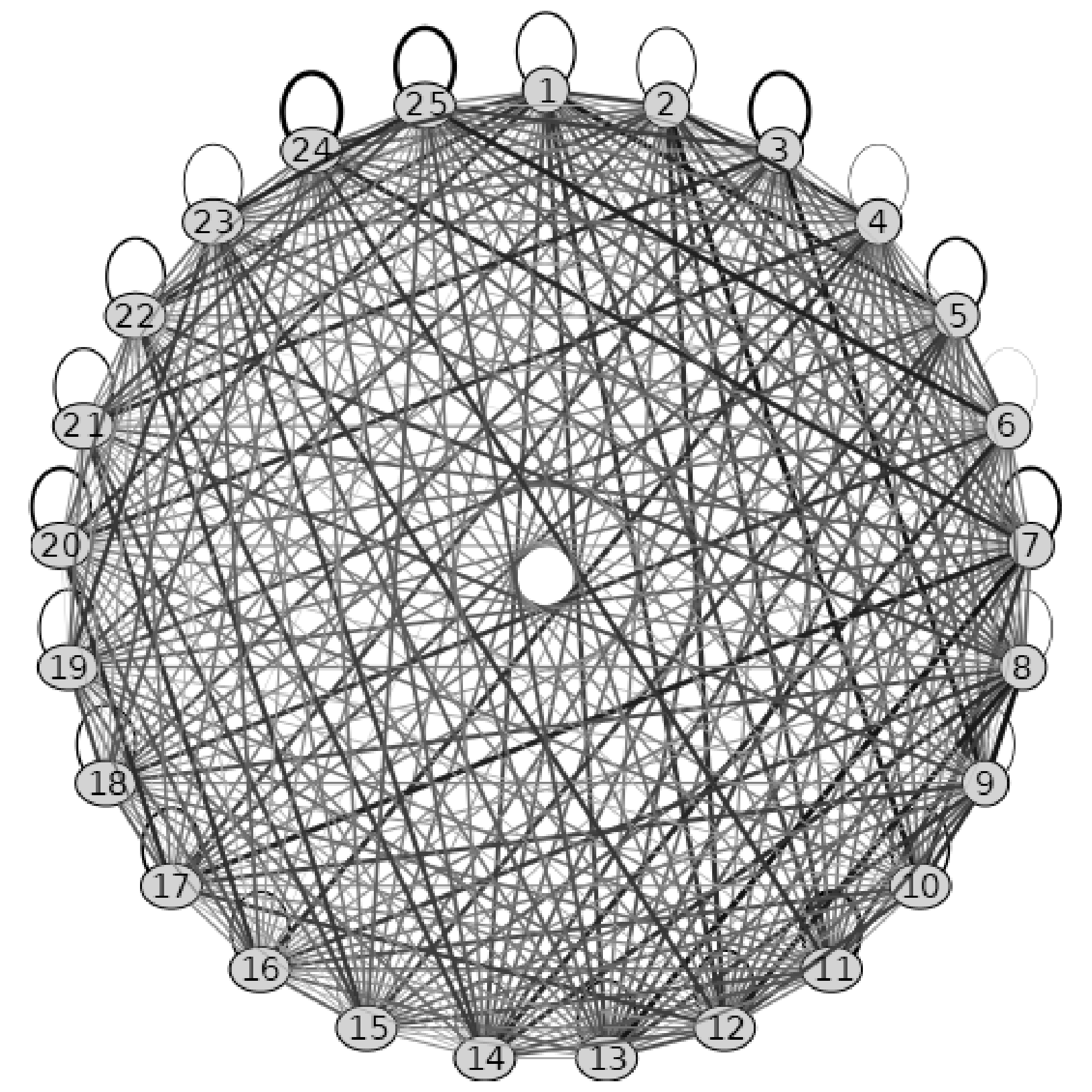} & 
    \includegraphics[width=0.36\textwidth]{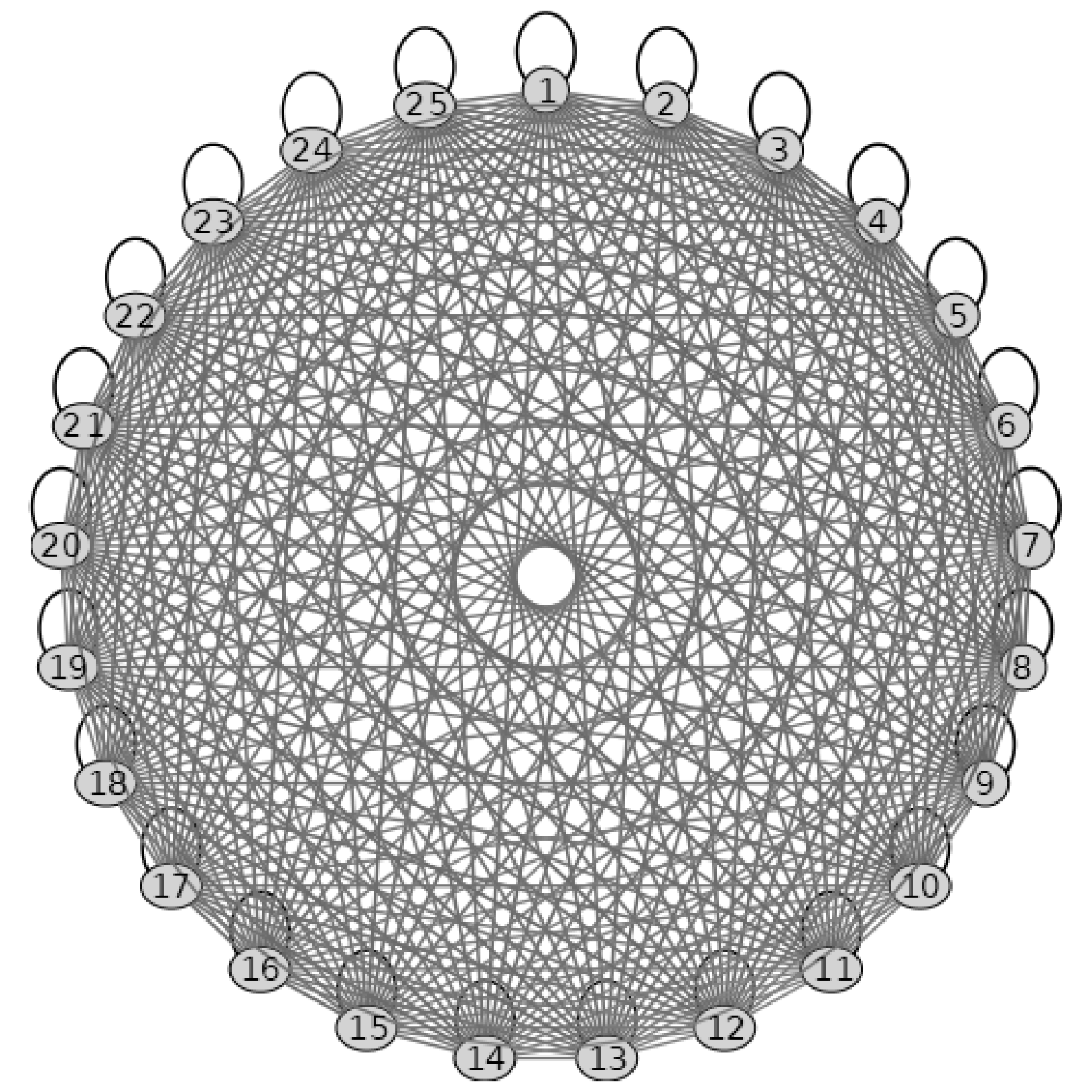}
    \end{tabular}
    \caption{Dynamics on the initial random graph in (a) obtained by applying 
    Eq.~(\ref{eqn:example1_localrule}).}
    \label{fig:example2}
\end{figure}


Data from the corresponding weight matrices which model these dynamics are provided in Eq.~(\ref{eqn:example2_M0}) through Eq.~(\ref{eqn:example2_M100}).
Because the weight matrices $M\in[0,1]^{25 \times 25}$ are rather large, we provide only submatrices of the weight matrices, denoted $M_{n}^\text{sub}$ where $n$ indicates the number of applications of the global rule. 
We consider the edges that connect the $21^\text{st}$ through $24^\text{th}$ nodes to the $10^\text{th}$ through $13^\text{th}$ nodes, as indicated by the labels. The dynamics that are visually suggested by the darkness of edges in Figure~\ref{fig:example2} are made concrete in Eq.~(\ref{eqn:example2_M0}) to Eq.~(\ref{eqn:example2_M100}). For example, the weight matrix indicates that the edge connecting the $23^\text{rd}$ and $11^\text{th}$ nodes is weak initially, but becomes very strong after the first application of the global rule.

\begin{align}
(a)\hspace{3mm} M_0^\text{sub} \hspace{3mm} &=\hspace{5mm} \bordermatrix{ & c_{21} & c_{22} & c_{23} &c_{24}\cr
                r_{10} & 0.1628 &  0.3862  & 0.9331 & 0.6953 \cr
                r_{11} & 0.4170  & 0.9254 & 0.2102 & 0.4118 \cr
                r_{12} & 0.2416 & 0.8147 & 0.4681 & 0.4712 \cr
                r_{13} & 0.7467 & 0.9784  & 0.7457 & 0.1173}
                \label{eqn:example2_M0}
\end{align}
\begin{align}
(b)\hspace{3mm} M_1^\text{sub} \hspace{3mm} &=\hspace{5mm} \bordermatrix{ & c_{21} & c_{22} & c_{23} &c_{24}\cr
                r_{10} & 0.3099 &  0.3116  & 0.4496 & 0.6663 \cr
                r_{11} & 0.1654  & 0.4020 & 0.9162 & 0.4758 \cr
                r_{12} & 0.5180 & 0.6072 & 0.4925 & 0.4443 \cr
                r_{13} & 0.3723 & 0.7820  & 0.6974 & 0.6015} \\[2mm]
(c)\hspace{3mm} M_2^\text{sub} \hspace{3mm} &=\hspace{5mm}  \bordermatrix{ & c_{21} & c_{22} & c_{23} &c_{24}\cr
                r_{10} & 0.4282 &  0.3683  & 0.2777 & 0.3830 \cr
                r_{11} & 0.2570  & 0.2559 & 0.4308 & 0.7773 \cr
                r_{12} & 0.2457 & 0.4587 & 0.7656 & 0.4872 \cr
                r_{13} & 0.4663 & 0.5034  & 0.6273 & 0.5616}\\[2mm]
(d)\hspace{3mm} M_{100}^\text{sub}\hspace{3mm}  &= \hspace{5mm} \bordermatrix{ & c_{21} & c_{22} & c_{23} &c_{24}\cr
                r_{10} & 0.5483 &  0.5282  & 0.5559 & 0.5296 \cr
                r_{11} & 0.5242  & 0.5522 & 0.5240 & 0.5281 \cr
                r_{12} & 0.5483 & 0.5180 & 0.5224 & 0.5185 \cr
                r_{13} & 0.5119 & 0.5165 & 0.5115 & 0.5067}
                \label{eqn:example2_M100}
\end{align}

By the $100^\text{th}$ application of the global rule, all edges in Figure~\ref{fig:example2}~(d) appear, at least visually, to be equal. The submatrix in Eq.~(\ref{eqn:example2_M100}) reveals that there are still some small differences between the probabilities, all of which hover around $0.5$. 
Performing more iterations reveals that the probabilities are in fact converging to a constant value.
Between 2500 and 3000 iterations are needed for the submatrix to have all entries equal to each other within four decimal places, with constant value $p=0.5101$.
The fact that the graph appears to settle at a constant graph may be an artifact
of our choice of periodic boundary conditions.

\section{Conclusions}

We have extended the notion of CCA to the case of probability measures. In this new framework the state space is the set $\mathcal{M}(\setofstates)$ which is composed by all probability measures defined on \setofstates. 
This extension is motivated by the idea of introducing uncertainty in cellular automata by acting on the status of each cell rather than on the transition map. This is the main difference between our approach and those currently existing in the literature.
We have proved several results which provide conditions for the existence of a limit point and the convergence of sequences of  configurations towards them. 
Using a simple example, we illustrate how CAMs may be used to model dynamics on random graphs and the link to symbolic dynamics.

{\color{black}
Here we have presented only a first introduction to our setting in which we have omitted many interesting results and extensions.
Indeed, future work will present further results relating to convergence, equicontinuity and stability. In future analysis we will also extend the ideas presented here by employing alternative metrics on the space $\MSL$. 
Moreover, many of the results presented above can also be extended to the case in which the coefficients $\lambda_j$ in the convex combination map are space-dependent, and this extension will also be pursued in future work.
Another extension omitted here is to consider CAMs in which we employ a fractal-type map $T$ in the local rule. 
}
Our novel setting of probability measures will also be extended to the more general setting of imprecise probabilities and set-valued probabilities.
We will also consider a generalized definition of CA in which the transition map is set-valued \cite{kunze2007contractive,la2009measure}.

Extensions arising from our choice of application are equally numerous.
{\color{black}
In particular, as illustrated by our example presented above, we believe that CAMs present a novel way to stimulate research using random graphs. 
A particular example, which we will pursue in future work, can be found in neural network thoery. 
An ideal network is characterized by a single output node and infinitely many input nodes. A single edge connects each of the (infinitely many) input nodes to the output node.
Associated with each edge is a probability measure which describes the weight intensity.
A particular configuration of this neural network can therefore be described by means of a configuration of a one-dimensional cellular automata. The training phase in which weights are updated corresponds to a specific dynamic on the corresponding cellular automata model. This setting can be easily extended to multi-layer neural networks which are modeled by means of higher dimensional cellular automata on probability measures.
}

Another future avenue includes the analysis of inverse problems for cellular automata on probability measures. Very briefly, given a configuration in \MSL an inverse problem consists of determining a cellular automaton that has this configuration as a global attracting point.

\begin{credits}
\subsubsection{Acknowledgements} This study was partially supported by
the HORIZON-MSCA-2022-SE-01 project 101131549 ``Application-driven Challenges for Automata Networks and Complex Systems (ACANCOS)'' and the ANR project ANR-24-CE48-0335-01 ``Ordinal Time Computations'' (OTC). Morevoer, we acknowledge the support from OTC project
of RISE Academy of Universit\'{e} C\^{o}te d'Azur.
\subsubsection{\discintname}
The authors have no competing interests to declare that are relevant to the content of this article. 
\end{credits}
%
%
%
 \bibliographystyle{splncs04}
 \bibliography{mybibliography}

\end{document}